%% file: main.tex
\documentclass[11pt]{article}
\usepackage{ifthen}
\usepackage{amsmath, amssymb}
\usepackage{cite}
\usepackage{url}
\usepackage{appendix}
\usepackage{graphicx}
\usepackage{epstopdf}
\usepackage[usenames,dvipsnames]{color}
\usepackage{algorithm}
\usepackage[noend]{algpseudocode}
\usepackage{multirow}
\usepackage{xspace}
\usepackage{wrapfig}
\usepackage{caption}
\usepackage{xspace}
\usepackage{mdwlist}
\usepackage[framemethod=tikz]{mdframed}
\usepackage{caption}
\usepackage[colorinlistoftodos]{todonotes}
\usepackage{xcolor}

\usepackage{enumitem}
\setitemize{itemsep=2pt, topsep=2pt,parsep=2pt,partopsep=2pt}
\definecolor{darkgreen}{rgb}{0,0.5,0}
\usepackage{hyperref}
\hypersetup{
    unicode=false,          % non-Latin characters in Acrobat’s bookmarks
    colorlinks=true,        % false: boxed links; true: colored links
    linkcolor=red,          % color of internal links (change box color with linkbordercolor)
    citecolor=darkgreen,        % color of links to bibliography
    filecolor=magenta,      % color of file links
    urlcolor=cyan           % color of external links
}
\usepackage[capitalize]{cleveref}
\newcommand{\FullOrShort}{full}

\ifthenelse{\equal{\FullOrShort}{full}}{
	
	  \newcommand{\fullOnly}[1]{#1}
	  \newcommand{\shortOnly}[1]{}
		
		\usepackage[letterpaper,margin=1.00in]{geometry}
		\usepackage{amsthm}
		\newtheorem{theorem}{Theorem}[section]
		\newtheorem{lemma}[theorem]{Lemma}

		\newtheorem{observation}[theorem]{Observation}
		\newtheorem{definition}[theorem]{Definition}
	}{

    \usepackage{times}
	  \newcommand{\fullOnly}[1]{}
	  \newcommand{\shortOnly}[1]{#1}
		
		\newtheorem{observation}[theorem]{Observation}
	}

\algnewcommand\algorithmicswitch{\textbf{switch}}
\algnewcommand\algorithmiccase{\textbf{case}}

\algdef{SE}[SWITCH]{Switch}{EndSwitch}[1]{\algorithmicswitch\ #1\ \algorithmicdo}{\algorithmicend\ \algorithmicswitch}%
\algdef{SE}[CASE]{Case}{EndCase}[1]{\algorithmiccase\ #1}{\algorithmicend\ \algorithmiccase}%
\algtext*{EndSwitch}%
\algtext*{EndCase}%

\newcommand{\eps}{\varepsilon}

\newcommand{\poly}{\operatorname{\text{{\rm poly}}}}
\newcommand{\Cost}{\operatorname{\text{{\rm cost}}}}

\DeclareMathAlphabet{\mathbfsf}{\encodingdefault}{\sfdefault}{bx}{n}

\newcommand{\local}{\ensuremath{\mathsf{LOCAL}}\xspace}
\newcommand{\congest}{\ensuremath{\mathsf{CONGEST}}\xspace}
\newcommand{\MST}{\ensuremath{\mathsf{MST}}\xspace}
\newcommand{\MCDS}{\ensuremath{\mathsf{MCDS}}\xspace}

\renewcommand{\paragraph}[1]{\vspace{0.15cm}\noindent {\bf #1}:}

\begin{document}

\date{}

\title{Near-Optimal Distributed Approximation of \\
Minimum-Weight Connected Dominating Set}

\shortOnly{
\author{
\vspace{-12pt}
  Mohsen Ghaffari
  %\small MIT\\
  %\texttt{\small ghaffari@mit.edu}
}
\institute{MIT\\ \email{ghaffari@mit.edu}}
}
\fullOnly{
\author{
  Mohsen Ghaffari \\
  \small MIT\\
  \texttt{\small ghaffari@mit.edu}
}
}
\maketitle

\shortOnly{\vspace{-15pt}}
\begin{abstract}
This paper presents a near-optimal distributed approximation algorithm for the minimum-weight connected dominating set (\MCDS) problem. We use the standard distributed message passing model called the \congest model in which in each round each node can send $\mathcal{O}(\log n)$ bits to each neighbor. The presented algorithm finds an $\mathcal{O}(\log n)$ approximation in $\tilde{\mathcal{O}}(D+\sqrt{n})$ rounds, where $D$ is the network diameter and $n$ is the number of nodes. 

\hspace{3ex} \MCDS is a classical $\mathsf{NP}$-hard problem and the achieved approximation factor $\mathcal{O}(\log n)$ is known to be optimal up to a constant factor, unless $\mathsf{P}=\mathsf{NP}$. Furthermore, the $\tilde{\mathcal{O}}(D+\sqrt{n})$ round complexity is known to be optimal modulo logarithmic factors (for any approximation), following [Das Sarma et al.---STOC'11].  
\end{abstract}

\section{Introduction and Related Work}
Connected dominating set (CDS) is one of the classical structures studied in graph optimization problems which also has deep roots in networked computation. For instance, CDSs have been used rather extensively in distributed algorithms for wireless networks (see e.g. \cite{chen2002approximating, alzoubi2002message, wu2001calculating, cheng2008connected, das1997routing, min2006improving, blum2005connected, cheng2003polynomial, alzoubi2002new, dai2004extended, wan2002distributed}), typically as a global-connectivity backbone. 

This paper investigates distributed algorithms for approximating \emph{minimum-weight connected dominating set} (\MCDS) while taking \emph{congestion} into account. We first take a closer look at what each of these terms means.

\vspace{-10pt}
%\vspace{40pt}

%the \emph{tour de force} in the study of distributed algorithms in the \congest model in the sense that, most of the main techniques for coping with congestion and also for proving lower bounds in the \congest model have been developed through studying \MST (see e.g. \cite{Kutten-Peleg, Garay-Kutten-Peleg, Elkin-2004, Peleg-Rubinovich-1999, DasSarma-11}) and also, many of the algorithms for other distributed problems use the \MST algorithm as a black-box tool (see e.g.\cite{DasSarma-11, distributed-cut, Thurimella}).

%\vspace{-10pt}
\subsection{A Closeup of $\mathbfsf{MCDS}$, in Contrast with $\mathbfsf{MST}$}
Given a graph $G=(V, E)$, a set $S\subseteq V$ is called a \emph{dominating set} if each node $v\notin S$ has a neighbor in $S$, and it is called a \emph{connected dominating set} (CDS) if the subgraph induced by $S$ is connected. \Cref{fig:CDS} shows an example. In the \emph{minim-weight CDS} (\MCDS) problem, each node has a weight and the objective is to find a CDS with the minimum total weight.

The \MCDS problem is often viewed as the node-weighted analogue of the \emph{minimum-weight spanning tree} (\MST) problem. Here, we recap this connection. The natural interpretation of the definition of CDS is that a CDS is a selection of \emph{nodes} that provides \emph{global-connectivity}---that is, any two nodes of the graph are connected via a path that its internal nodes are in the CDS. On the counterpart, a \emph{spanning tree} is a (minimal) selection of \emph{edges} that provides global-connectivity. In both cases, the problem of interest is to minimize the total weight needed for global-connectivity. In one case, each edge has a weight and the problem becomes \MST; in the other, each node has a weight and the problem becomes \MCDS.

\begin{figure}[t]
	\centering
		\includegraphics[width=0.6\textwidth]{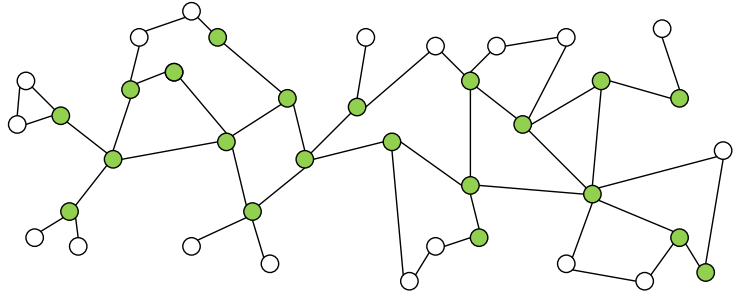}
	\caption{{\scriptsize The green nodes represent a connected dominating set (CDS) of the graph.}}
	\label{fig:CDS}
	\vspace{-15pt}
\end{figure}

%\medskip
Despite the seemingly analogous nature of the two problems, \MCDS turns out to be a significantly harder problem: The \MST problem can be computed sequentially in (almost) $\mathcal{O}(m)$ time, where $m$ is the number of edges. On the other hand, $\MCDS$ is $\mathsf{NP}$-hard \cite{Garey:1990}, and in fact, unless $\mathsf{P}=\mathsf{NP}$, no polynomial time algorithm can find any approximation better than $\Theta(\log n)$-factor for it (see \cite{Feige, Raz-Safra, Alon-Moshkovitiz-Safra}). Furthermore, the known sequential algorithms for $\mathcal{O}(\log n)$ approximation of \MCDS (see \cite{Guha-Khuller, Guha-Khuller-2}) have unspecified polynomial time complexity, which are at least $\Theta(n^3)$. 
%\medskip

\subsection{Congestion in Distributed Algorithms}
Two central issues in distributed computing are \emph{locality} and \emph{congestion}\cite{Peleg:2000}. Classically, locality has received more attention and most graph problems were studied in the \local model, where congestion is ignored and messages can have unbounded size. The recent years have seen a surge in focus on understanding the effect of congestion in graph problems (see e.g., \cite{DasSarma-Walk-2009, DasSarma-Walk-2010, Nanongkai-walk, Danupon-paths, Lenzen-PattShamir, distributed-cut, Frischknecht-Diameter-2012, Holzer-Paths-2012, DasSarma-11}). The standard distributed model that takes congestion into account is called \congest\cite{Peleg:2000}, where in each round, each node can send $B$ bits to each of its neighbors, and normally one assumes $B=\mathcal{O}(\log n)$. \fullOnly{It is well-known that even the easiest of problems in the \local model can become challenging in the \congest model.}
The pioneering problem in the study of the \congest model was \MST: A beautiful line of work shows that \MST can be solved in $\mathcal{O}(D+\sqrt{n}\log^*n)$ rounds\cite{Garay-Kutten-Peleg, Kutten-Peleg} and that this is (existentially) optimal modulo logarithmic factors\cite{DasSarma-11, Elkin-2004, Peleg-Rubinovich-1999}, and a similar lower bound also applies to many other distributed graph problems~\cite{DasSarma-11}\fullOnly{\footnote{For the reader interested in distributed (approximation) algorithms while considering congestion, the author recommends reading \cite{DasSarma-11} and the prior work on that thread, e.g., \cite{Elkin-2004, Peleg-Rubinovich-1999}.}}. Since then, achieving an $\tilde{\mathcal{O}}(D+\sqrt{n})$ round complexity is viewed as sort of a golden standard for (non-local) problems in the \congest model. The area is quite active and in the last couple of years, a few classical graph optimization problems (which are in $\mathsf{P}$) are shown to have approximation matching this standard or getting close to it: some distance-related problems such as shortest-path approximations~\cite{Lenzen-PattShamir, Danupon-paths} or diameter and girth approximations\cite{Holzer-Paths-2012}, and minimum-cut approximation\cite{distributed-cut}.

\vspace{-5pt}
\subsection{Result}
\vspace{-5pt}
The contribution of this paper is to show that in the \congest model, \MCDS can be solved---that is, approximated optimally---in a time close to that of \MST.
\fullOnly{\smallskip}
\begin{mdframed}[hidealllines=true,backgroundcolor=gray!35]
%\begin{center}
\fullOnly{\vspace{-4pt}}
\shortOnly{\vspace{2pt}}
\begin{theorem}
There is a randomized distributed algorithm in the \congest model that, with high probability, finds an $\mathcal{O}(\log n)$ approximation of the minimum-weight connected dominating set, using $\tilde{\mathcal{O}}(D+\sqrt{n})$ rounds. 
\end{theorem}
\vspace{2pt}
%\end{center}
\end{mdframed}
\fullOnly{\smallskip}

This algorithm is (near) optimal in both round complexity and approximation factor: Using techniques of \cite{DasSarma-11}, one can reduce the \emph{two-party set-disjointness communication complexity} problem on $\Theta(\sqrt{n})$-bit inputs to \MCDS, proving that the round complexity is optimal, up to logarithmic factors, for any approximation (see Appendix B \shortOnly{in the full version}). As mentioned above, the $\mathcal{O}(\log n)$ approximation factor is known to be optimal up to a constant factor, unless $\mathsf{P}=\mathsf{NP}$, assuming that nodes can only perform polynomial-time computations. Note that this assumption is usual, see e.g. \cite{Dubhashi-03, Jia-Rajmohan-Suel, Kuhn-Wattenhofer-03}. 

\subsection{Other Related Work} 
To the best of our knowledge, no efficient algorithm was known before for \MCDS in the \congest model. Notice that in the \local model, \MCDS boils down to a triviality and is thus never addressed in the literature: it is folklore\shortOnly{\footnote{In a cycle with $2D$ nodes, nodes need to learn the weight of the node at the opposite side, which is $D$ hops away and requires $D$ rounds.}}\fullOnly{\footnote{On one hand, $D$ rounds is enough for learning the whole graph. On the other, $D$ rounds is necessary for guaranteeing any approximation factor $\alpha$. Consider a cycle with with $2D$ nodes where two nodes $v$ and $u$ are at distance $D$. For each of $v$ and $u$, assign a random weight in $\{n^2, n^2\alpha+1\}$ and give weight $1$ to each other node. For the CDS to $\alpha$-factor optimal, the following should hold: if one of $v$ and $u$ has cost $n^2\alpha+1$, then before joining the CDS, it needs to make sure that the other does not have weight $n^2$. This requires $D$ rounds.}} that in this model, $D$ rounds is both necessary and sufficient for any approximation of \MCDS. However, a special case of \MCDS is interesting in the \local model; the so-called ``unweighted case" where all nodes have equal weight. Although, the unweighted-case has a significantly different nature as it makes the problem ``\emph{local}": Dubhashi et al.\cite{Dubhashi-03} present a nice and simple $\mathcal{O}(\log n)$ approximation for the unweighted-case algorithm which uses $\mathcal{O}(\log^2 n)$ rounds of the \local model. To our knowledge, the unweighted case has not been addressed in the \congest model, but we briefly comment in \fullOnly{\Cref{app:unweighted}}\shortOnly{Appendix A of the full version} that one can solve it in $\mathcal{O}(\log^2 n)$ rounds of the \congest model as well, by combining the dominating set approximation of Jia et al.\cite{Jia-Rajmohan-Suel} with the \emph{linear skeleton} of Pettie\cite{Pettie-Skeleton} and a simple trick for handling congestion. Another problem which has a name resembling \MCDS is the \emph{minimum-weight dominating set} ($\mathsf{MDS}$) problem. However, $\mathsf{MDS}$ is also quite different from \MCDS as the former is ``\emph{local}", even in the weighted case and the \congest model: an $\mathcal{O}(\log n)$ factor approximation can be found in $\mathcal{O}(\log^2 n)$ rounds\cite{Jia-Rajmohan-Suel, Kuhn-Wattenhofer-03} (see also  \cite{Kuhn-locality:2004}).  

%\subsection{Other Related Work}
%\dots 
%
%Algorithm \cite{Jia-Rajmohan-Suel, Kuhn-Wattenhofer-03, Dubhashi-03, Guha-Khuller}
%
%Approximation Hardness \cite{Feige, Raz-Safra}, 
%
%Time Lower Bound \cite{DasSarma-11}

\section{Preliminaries}
%\subsection{Model, Notations, and Problem Statement}
\paragraph{Distributed Model} As stated above, we use the \congest model: communication between nodes happens in lock-step rounds where in each round, one $B$-bits message can be sent on each direction of each edge, and we particularly focus on the standard case of $B=\mathcal{O}(\log n)$. The only global knowledge assumed is that nodes know an upper bound $N=\poly(n)$ on $n$. \fullOnly{We assume each node has a unique id with $\mathcal{O}(\log n)$ bits, although this is not critical as each node simply picking a random id in $\{0,1\}^{4\log N}$ would ensure uniqueness of ids, with high probability.} We use the phrase \emph{with high probability} (w.h.p.) to indicate a probability being at least $1-\frac{1}{n^\beta}$, for a constant $\beta\geq 2$. 

\paragraph{Notations and basic definitions}
We work with an undirected graph $G=(V,E)$, $n=|V|$, and for each vertex $v\in V$, $c(v)$ denotes the weight (i.e., cost) of node $v$. Throughout the paper, we will use the words \emph{cost} and \emph{weight} interchangeably. For each subset $T \subseteq V$, we define $\Cost(T) = \sum_{v\in T} c(v)$. We assume the weights are at most polynomial in $n$, so each weight can fit in one message (such assumptions are usual, e.g. \cite{Garay-Kutten-Peleg}). We use notation $\mathsf{OPT}$ to denote the CDS with the minimum cost. Also, for convenience and when it does not lead to any ambiguity, we sometimes use $\mathsf{OPT}$ to refer to the cost of the optimal CDS.

\paragraph{Problem Statement} Initially, each node $v$ knows only its own weight $c(v)$. The objective is to find a set $S$ in a distributed fashion---that is, each node $v$ will need to output whether $v\in S$ or not---such that $\Cost(S)=\mathcal{O}(\mathsf{OPT} \cdot \log n )$.

\paragraph{A Basic Tool (Thurimella's algorithm)}
\label{subsec:Thurimella}
%\vspace{-10pt}
A basic tool that we frequently use is a \emph{connected component identification algorithm} presented by Thurimella~\cite{Thurimella}, which itself is a simple application of the \MST algorithm of Kutten and Peleg\cite{Kutten-Peleg}.
Given a subgraph $H=(V, E')$ of the main network graph $G=(V, E)$, this algorithm identifies the connected components of $H$ by giving a label $\ell(v)$ to each $v$ such that $\ell(v)=\ell(u)$ if an only if $v$ and $u$ are in the same connected component of $H$. This algorithm uses $\mathcal{O}(D+\sqrt{n}\log^* n)$ rounds of the \congest model. It is easy to see that the same strategy can be adapted to solve the following problems also in $\mathcal{O}(D+\sqrt{n}\log^* n)$ rounds. Suppose each node $v$ has an input $x(v)$. For each node $v$, which is in a component $\mathcal{C}$ of $H$, we can make $\ell(v)$ be equal to: (A) the \emph{maximum} value $x(u)$ for nodes $u \in \mathcal{C}$ in the connected component of $v$, or (B) the \emph{list of $k=\mathcal{O}(1)$ largest} values $x(u)$ for nodes $u \in \mathcal{C}$, or (C) the \emph{summation} of values $x(u)$ for nodes $u \in \mathcal{C}$.

\vspace{-5pt}
\section{The Algorithm for \MCDS}
%\vspace{-10pt}
%In this section, we present our distributed algorithm for $\mathcal{O}(\log n)$ approximation of \MCDS, in $\tilde{\mathcal{O}}(D+\sqrt{n})$ rounds of the \congest model. 
\vspace{-5pt}
\subsection{The Outline}
\vspace{-5pt}
The top-level view of the approach is as follows: We start by using the $\mathcal{O}(\log^2 n)$ rounds algorithm of \cite{Jia-Rajmohan-Suel} to find a dominating set $S$ with cost $\mathcal{O}(\log n \cdot \mathsf{OPT})$. The challenge is in adding enough nodes to connect the dominating set, while spending extra cost of $\mathcal{O}(\log n \cdot \mathsf{OPT})$. We achieve connectivity in $\mathcal{O}(\log n)$ phases. In each phase, we add some nodes to set $S$ so that we reduce the number of connected components of $S$ by a constant factor, while spending a cost of $\mathcal{O}(\mathsf{OPT})$. After $\mathcal{O}(\log n)$ phases, the number of connected components goes down to $1$, meaning that we have achieved connectivity.  Each phase uses $\tilde{\mathcal{O}}(D+\sqrt{n})$ rounds of the \congest model. What remains is to explain how a phase works. 

The reader might recall that such ``component-growing" approaches are typical in the \MST algorithms, e.g., \cite{Kutten-Peleg, Garay-Kutten-Peleg}. While in \MST, the choice of the edge to be added to each component is clear (\emph{the lightest outgoing edge}), the choice of the nodes to be added in \MCDS is not clear (and in fact can be shown to be an $\mathsf{NP}$-hard problem, itself).

The problem addressed in one phase can be formally recapped as follows (the reader might find the illustration in \Cref{fig:BaseView} helpful here): We are given a dominating subset $S\subseteq V$ and the objective is to find a subset $S' \subseteq V\setminus S$ with $\Cost(S') =\mathcal{O}(\mathsf{OPT})$ such that the following condition is satisfied. Let $\mathcal{F}$ be the set of subsets of $S$ such that each $\mathcal{C}\in \mathcal{F}$ is a connected component of $G[S]$. Call a connected component $\mathcal{C}\in \mathcal{F}$ \emph{satisfied} if in $G[S\cup S']$, $\mathcal{C}$ is connected to at least one other component $\mathcal{C}'\in \mathcal{F}$. We want $S'$ to be such that at least half of the connected components of $G[S]$ are satisfied. Note that if this happens, then the number of connected components goes down by a $3/4$ factor. To refer to the nodes easier, we assume that all nodes that are in $S$ at the start of the phase are colored \emph{green} and all the other nodes are \emph{white}, initially. During the phase, some white nodes will become gray meaning that they joined $S'$.

Before moving on to the algorithm, we emphasize two key points:

\begin{itemize}
\item[(1)] It is critical to seek satisfying only a constant fraction of the components of $G[S]$. Using a simple reduction from the set cover problem, it can be shown that satisfying all components might require a cost $\mathcal{O}(\mathsf{OPT} \log n)$ for a phase. Then, at least in the straightforward analysis, the overall approximation factor would become $\mathcal{O}(\log^2 n)$.
\item[(2)] In each phase, we \emph{freeze} the set of components $\mathcal{F}$ of $G[S]$. That is, although we continuously add nodes to the CDS and thus the components grow, we will not try to satisfy the newly formed components. We keep track of whether a component $\mathcal{C}\in \mathcal{F}$ is satisfied and the satisfied ones become ``inactive" for the rest of the phase, meaning that we will not try to satisfy them again. However, satisfied components will be used in satisfying the others.
\end{itemize}
\begin{figure}[t]
	\centering
		\includegraphics[width=0.60\textwidth]{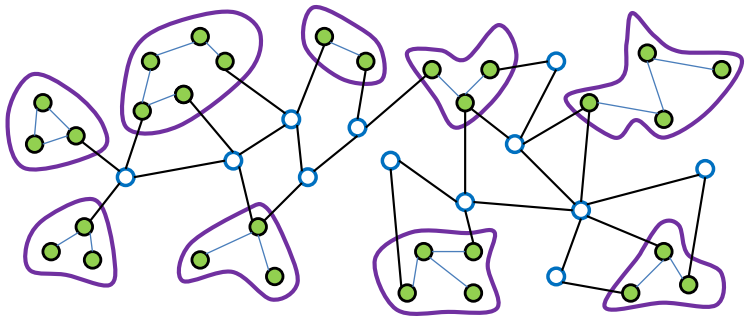}
	\caption{{\scriptsize An example scenario at the start of a phase. Green nodes indicate those in $S$ and white nodes are $V\setminus S$. Unrelated nodes and edges are discarded from the picture.}}
	\label{fig:BaseView}
	\vspace{-10pt}
\end{figure}

\subsection{A High-level View of the Algorithm for One Phase}
\begin{wrapfigure}{r}{0.30\textwidth}
  \shortOnly{\vspace{-40pt}}
	\begin{center}
    \includegraphics[width=35 mm]{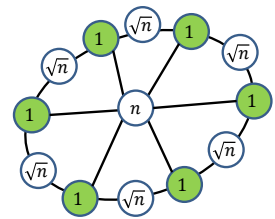}
  \end{center}
  \shortOnly{\vspace{-15pt}}
	\caption{{\scriptsize The naive approach}}
	\label{fig:naive}
	\shortOnly{\vspace{-20pt}}
\end{wrapfigure}
Note that since $S$ is a dominating set, $\mathcal{C}\in \mathcal{F}$ is satisfied iff there exist one or two nodes that connect $\mathcal{C}$ to another component $\mathcal{C'} \in \mathcal{F}$. That is, either there is a node $v$ such that path $\mathcal{C}$-$v$-$\mathcal{C}'$ connects component $\mathcal{C}$ to component $\mathcal{C}'$ or there are two adjacent nodes $v$ and $w$ such that path $\mathcal{C}$-$v$-$w$-$\mathcal{C}'$ does that. Having this in mind, and motivated by the solution for the unweighted case\cite{Dubhashi-03}, a naive approach would be that, for each component $\mathcal{C}$, we pick one or two nodes---with smallest total weight---that connect $\mathcal{C}$ to another component, and we do this for each component $\mathcal{C}$ independently. However, in the weighted case, this naive idea would perform terribly. To see why, let us consider a simple example (see \Cref{fig:naive}): take a cycle with $n-1$ nodes where every other node has weight $1$ and the others have weight $\sqrt{n}$, and then add one additional node at the center with weight $n$, which is connected to all weight-$1$ nodes. Clearly, the set of weight-$1$ nodes gives us an optimal dominating set. However, naively connecting this dominating set following the above approach would make us include at least half of the $\sqrt{n}$-weight nodes, leading to overall weight of $\Theta(n\sqrt{n})$. On the other hand, simply adding the center node $s$ to the dominating set would provide us with a CDS of weight $\mathcal{O}(n)$.

Inspired by this simple example, we view \emph{stars} as the key elements of optimization (instead of $2$ or $3$ hop paths). We next define what we mean by a star and outline how we use it. We note that the concept is also similar to the notion of \emph{spiders} used in~\cite{Klein-Ravi} for the node-weighted Steiner trees problem.   

\begin{definition}(\textbf{Stars}) A star $X$ is simply a set of white nodes with a \emph{center} $s\in X$ such that each non-center node in the star is connected to the center $s$. Naturally, we say a star $X$ \emph{satisfies} an active component $\mathcal{C}\in \mathcal{F}$ if adding this star to $S'$---that is, coloring its nodes gray---would connect $\mathcal{C}$ to some other component and thus make it satisfied. Let $\Phi(X)$ be the set of unsatisfied components in $\mathcal{F}$ that would be satisfied by $X$. We say a star is useless if $\Phi(X)=\emptyset$. The cost of a star $X$ is $\Cost(X)=\sum_{w\in X} c(w)$ and its efficiency is $\rho(X)=\frac{|\Phi(X)|}{\Cost(X)}$. We say $X$ is $\rho'$-efficient if $\rho(X)\geq \rho'$. 
\end{definition}

In \Cref{fig:naive}, each white node is one star, the center has efficiency $\Theta(1)$ and every other star has efficiency $\Theta(1/\sqrt{n})$. Notice that in general, different stars might intersect and even a white node $v$ might be the center of up to $2^{\Theta(n)}$ different stars.

\begin{figure}[t]
	\centering
		\includegraphics[width=0.70\textwidth]{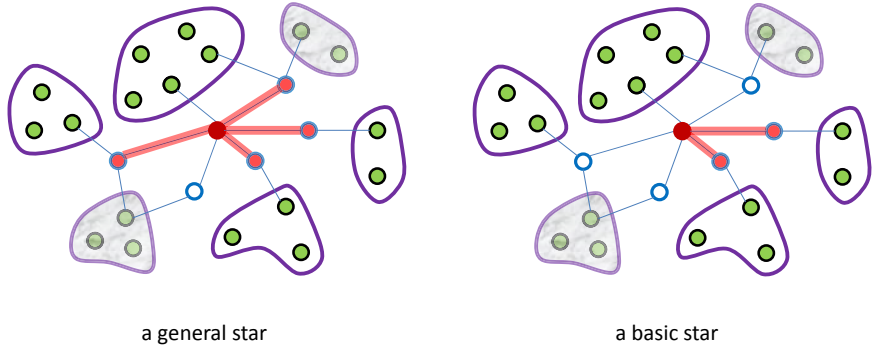}
		\vspace{-10pt}
	\caption{{\scriptsize A basic-star. The opaque components indicate those that are already satisfied and thus deactivated. Two legs of the general star (colored red, on the left) are discarded in the basic-star (colored red, on the right), as each of them forms a useful star, meaning that the leg itself can satisfy at least one active component.}}
	\label{fig:Star}
		\vspace{-15pt}
\end{figure}

\paragraph{The general plan (while ignoring some difficulties)} We greedily\fullOnly{\footnote{The greedy approaches are typically standard in solving \MCDS or other problems similar in nature. Furthermore, often the notion of \emph{efficiency} as explained above or some variant of it is the base of picking the next good move, in these greedy approaches. See e.g. \cite{Guha-Khuller, Guha-Khuller-2, Klein-Ravi, Berger-Rompel-Shor}.}} add stars to the gray nodes. That is, we pick a star that has the maximum efficiency and color its nodes gray. It can be shown that this greedy idea would satisfy half of components using cost only $\mathcal{\mathcal{O}}(\mathsf{OPT})$. However, clearly adding stars one by one would be too slow. Instead we adopt a nice and natural technique due to Berger et al.\cite{Berger-Rompel-Shor} which by now has become a standard trick for speeding up greedy approaches via parallelizing their steps. The key point is, stars that have efficiency within a constant factor of the max-efficiency are essentially as good as the max-efficient star and hence, we can add those as well. The only catch is, one needs to make sure that adding many stars simultaneously does not lead to (too much) \emph{double counting} in the efficiency calculations. In other words, if there are many stars that try to satisfy the same small set of components, even if each of these stars is very efficient, adding all of them is not a good idea. The remedy is to probabilistically add stars while the probabilities are chosen such that not too many selected stars try to satisfy one component. 

While this general outline roughly explains what we will do, the plan faces a number of critical issues. We next briefly hint at two of these challenges and present the definitions that we use in handling them.    

\paragraph{Challenge 1}The first step in the above outline is to compute (or approximate) the efficiency of the max-efficient star. Doing this for the general class of stars turns out the be a hard problem in the \congest model. Note that for a white node $v$ to find (or approximate) the most-efficient star centered on it, $v$ would need to know which components are adjacent to each of its white neighbors. As each white node might be adjacent to many components, this is like learning the $2$-neighborhood of $v$ and appears to be intrinsically slow in the \congest model. Instead, we will focus on a special form of stars, which we call \emph{basic-stars} and explain next. \Cref{fig:Star} shows an example.

\begin{definition}(\textbf{Basic-Stars})\label{def:basic-star} Call a white node $u$ \emph{self-sufficient} if $u$ is adjacent to two or more components, at least one of which is not satisfied. A star $X$ is called \emph{basic} if for each non-center node $w\in X$, $w$ is not self-sufficient. That is, the star $X'=\{w\}$ is useless.
\end{definition}

We argue later that, considering only the basic-stars will be sufficient for our purposes (sacrificing only a constant factor in the approximation quality) and that we can indeed evaluate the max-efficiency of the basic-stars.

%With this change, each white node needs to only report at most one component, if it has an unsatisfied adjacent component, or report that it is adjacent to only satisfied components.

\paragraph{Challenge 2} The other issue, which is a bit more subtle but in fact significantly more problematic, is as follows: as we color some white nodes gray, some components grow and thus, the efficiencies of the stars change. For instances, a useless star $X=\{v\}$ might now become useful--e.g., it gets connected to a satisfied component $\mathcal{C}'$ via a node $u$ that just got colored gray, and $X$ can now satisfy an adjacent unsatisfied component $\mathcal{C}$ by connecting it to $\mathcal{C}'$. Another example, which is rooted also in the congestion related issues, is as follows: During our algorithm, to be able to cope with communication issues, each white node $v$ will work actively on only one max-efficient basic-star centered on $v$. But, $v$ might be the center of many such stars and even if one of them looses the efficiency after this iteration, another max-efficient star which existed before might be now considered actively by $v$.

We note that, if there were no such \emph{``new-stars"} issues, we could use here standard methods such as (a modification of) the LP relaxation based technique of Kuhn and Wattenhofer~\cite{Kuhn-Wattenhofer-03}. However, these changes break that approach and it is not even clear how to formulate the problem as an LP (or even a convex optimization problem, for that matter).

If not controlled, these changes in the stars can slow down our plan significantly. For example, if for a given almost-maximum efficiency $\tilde{\rho}$, in each iteration a small number of $\tilde{\rho}$-efficient new basic-stars are considered actively, we will have to spend some time on these stars but as the result, we would satisfy only very few components, which would become prohibitively slow. To remedy this, when coloring stars gray, we will do it for certain types of $\tilde{\rho}$-efficient basic-stars, which we define next, and after that, we do some clean up work to remove the new $\tilde{\rho}$-efficient basic-stars that would be considered actively later on. 

\begin{wrapfigure}{r}{0.34\textwidth}
  \vspace{-10pt}
	\shortOnly{\vspace{-20pt}}
	\begin{center}
    \includegraphics[width=40 mm]{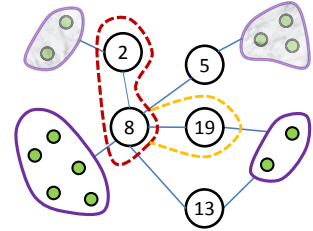}
  \end{center}
  \vspace{-10pt}
	\caption{{\scriptsize A $0.1$-augmented basic-star is indicated with the dashed lines; the red part is a minimal $0.1$-efficient basic-star and the orange part is a good auxiliary leg.}}
	\label{fig:augmented}
	\vspace{-20pt}
	\shortOnly{\vspace{-5pt}}
\end{wrapfigure}
\begin{definition}(\textbf{$\rho^*$-Augmented Basic-Stars})\label{def:aug-star}  A $\rho^*$-efficient basic-star $X$ centered on node $v\in X$ is called $\rho^*$-minimal if for any other star $X'\subset X$ centered on $v$, we have, $\rho(X') < \rho^*$. For a $\rho^*$-minimal basic-star $X$ centered on $v$, a \emph{good auxiliary-leg} is a white node $u \notin X$ that is adjacent to $v$ and furthermore, the following conditions are satisfied: $u$ is adjacent to only one component $\mathcal{C} \in \mathcal{F}$, component $\mathcal{C}$ is not satisfied and it is not adjacent to $X$, and we have $\Cost(u) \leq 2/\rho^*$. 
A \emph{$\rho^*$-Augmented Basic-Star} $X'$ is one that can be derived by (one-by-one)\fullOnly{\footnote{This has to be done one-by-one as adding one good auxiliary leg might make the star adjacent to a component $\mathcal{C}$ and then, no other white node adjacent to $\mathcal{C}$ can be a good auxiliary-leg.}} adding to $\rho^*$-minimal basic-star $X$ all good auxiliary-legs adjacent to its center.  
\end{definition}

An example is shown in \Cref{fig:augmented}. The actual reasoning for why this definition is good is somewhat subtle to be explained intuitively. A very rough version is as follows: after coloring some $\rho^*$-augmented basic-stars gray, by just handling the nodes which each have cost at most $1/\rho^*$ (in a step we call clean up), we will be able to remove any new $\rho^*$-augmented basic-star. The point should become clear after seeing the algorithm \fullOnly{(and Lemmas \ref{lem:cleanup} and \ref{lem:decreasing-degrees} in \Cref{app:analysis})}. 

%Before moving on, we state two simple facts.

\begin{observation}
Each $\rho^*$-Augmented Basic-Star $X$ has efficiency $\rho(X)\geq \frac{\rho^*}{2}$. Furthermore, if a $\rho^*$-Augmented Basic-Stars $X$ contains a white node $w$, then all unsatisfied components adjacent to $w$ get satisfied by $X$.
\end{observation}

%\begin{proof} Consider an unsatisfied component $\mathcal{C}$ adjacent to $w$. Since $X$ is not useless, it is adjacent to at least one other component $\mathcal{C}'\neq \mathcal{C}$ (which might be satisfied or unsatisfied). In either case, coloring nodes of $X$ gray would connect $\mathcal{C}$ to $\mathcal{C}'$, thus satisfying it.
%\end{proof}

%\begin{observation}For each $\rho^* >0$, if a $\rho^*$-Augmented Basic-Stars $X$ contains a white node $w$, either $w$ is self-sufficient, in which case it is the center of $X$, or $w$ is adjacent to at most one unsatisfied component $\mathcal{C}$, in which case, any useful star containing $w$ is satisfying $\mathcal{C}$ through $w$. If $w$ is the center, it knows which unsatisfied components are connected to $X$ only via $w$ and center $w$ will be responsible for only these components.
%\end{observation}

%Having these definitions, now we explain the algorithm.
\subsection{The Algorithm For One Phase}
\shortOnly{\vspace{-5pt}}
The objective of the algorithm is to satisfy at least half of the components, using a cost $\mathcal{O}(\mathsf{OPT})$, and in $\mathcal{O}((D+\sqrt{n}\log^* n)\log^3 n)$ rounds.
Throughout the phase, each non-white node will keep track of whether its component in $\mathcal{F}$ is satisfied or not. Let $N=|\mathcal{F}|$ and also, make all nodes know $N$ by running Thurimella's connected component identification at the start of the phase and then globally gathering the number of components. 

While at least $\lfloor N/2\rfloor$ components in $\mathcal{F}$ remain unsatisfied, we repeat the following iteration, which has $8$ steps---$\mathcal{S}1$ to $\mathcal{S}8$---and each step uses $\mathcal{O}(D+\sqrt{n}\log^* n)$ rounds: 
\begin{itemize}
\item[\textbf{($\mathcal{S}1$)}] We first use Thurimella's algorithm (see \Cref{subsec:Thurimella}) to identify the connected components of non-white nodes and also to find out whether each component is satisfied (i.e. if it contains a gray node). These take $\mathcal{O}(D+\sqrt{n}\log^* n)$ rounds. Each non-white node broadcasts its component id and whether its component is satisfied to all neighbors. We also find the total number of unsatisfied connected components and if it is less than $N/2$, we call this phase finished and start the next phase.

\item[\textbf{($\mathcal{S}2$)}] We now find the globally-maximum efficiency $\rho^*$ of the basic-stars. 

They key part is to compute the efficiency of the most-efficient basic-star centered on each white node. After that, the global-maximum can be found in $\mathcal{O}(D)$ rounds easily. We first use one round of message exchanges between the white nodes so that each white node knows all the basic-stars it centers.  

\hspace{3ex} Each white node $v$ does as follows: if $v$ is adjacent to only one component (satisfied or unsatisfied), it sends the id of this component, its satisfied/unsatisfied status and $v_{id}$ to its neighbors. If $v$ is adjacent to two or more components, but all of them are satisfied, then $v$ sends a message to its neighbors containing $v_{id}$ and an indicator message ``\emph{all-satisfied}". If $v$ is adjacent to two or more components, at least one of which is unsatisfied, then $v$ does not send any message. This is because, by \Cref{def:basic-star}, node $v$ is \emph{self-sufficient} and it thus can be only in basic-stars centered on $v$. At the end of this round, each white node $v$ has received some messages from its white neighbors. These messages contain all the information needed for forming all the basic-stars centered on $v$ and calculating their efficiency. Node $v$ finds the most-efficient of these basic-stars. It is easy to see that this can indeed be done in polynomial-time local computation\fullOnly{\footnote{For that, node $v$ first adds itself to the basic-star. Then, it discards any adjacent white node $u$ for which the only unsatisfied component adjacent to $u$ is also adjacent to $v$. Then, $v$ sorts all the remaining white-neighbors $u_1, u_2, ..., u_\ell$---from which it received a message---by increasing cost order. It then adds $u_i$-s one by one to its basic-star, as long as each new addition increases the efficiency. Since each white-neighbor is adjacent to at most one unsatisfied component, it is easy to see that this indeed finds the maximum efficiency.}}.
We emphasize that the basic-stars found in this step are not important and the only thing that we want is to find the globally-maximum efficiency $\rho^*$. 

\item[\textbf{($\mathcal{S}3$)}] Let $\tilde{\rho}= 2^ {\lfloor\log_{2}{\rho^*}\rfloor}$, i.e., $\tilde{\rho}$ is equal to $\rho^*$ rounded down to the closest power of $2$. We pick at most one $\tilde{\rho}$-augmented basic-star $X^i_v$  (see \Cref{def:aug-star}) centered on each white node $v$, where $i$ is the iteration number. 

\hspace{3ex} We reuse the messages exchanged in the previous step. First, each white node $v$ finds a minimal $\tilde{\rho}$-efficient basic-star centered on $v$, if there is one. Call this the \emph{core-star of $v$}. Then, $v$ adds to this core-star any good auxiliary-legs available (one by one), to find its $\rho^*$-augmented basic-star $X^i_v$. This is the only star centered on $v$ that will be considered for the rest of this iteration. 
Thus, at most one star $X^i_v$ centered on each white node $v$ remains active for the rest of iteration $i$. Note that all active remaining stars are $\tilde{\rho}/2$-efficient. 
%However, clearly not all $\tilde{\rho}/2$-efficient basic-stars remain active.

\hspace{3ex} For each active-remaining star $X^i_v$ and each unsatisfied component $\mathcal{C}$ it satisfies, the center $v$ elects one of the white nodes of the star to be \emph{responsible for communicating}\fullOnly{\footnote{Note that each star might have many nodes that are adjacent to an unsatisfied component $\mathcal{C}$. As this would be problematic for our communication purposes, we avoid this by making only one node in the star responsible for each unsatisfied adjacent component.}} with $\mathcal{C}$. If $\mathcal{C}$ has at least one non-center neighbor in $X^i_v$, then one such non-center node $u$ (selected arbitrarily) is called \emph{responsible for communicating with} $\mathcal{C}$. Otherwise, the center $v$ is responsible\footnote{Since any white node $u$ that is not self-sufficient is adjacent to at most one unsatisfied component, in any basic-star that contains $u$, node $u$ can be responsible only for this one unsatisfied adjacent component. On the other hand, if $v$ is self-sufficient, it will be only in one star $X^i_v$.} for communicating with $\mathcal{C}$.

\item[\textbf{($\mathcal{S}4$)}] For each unsatisfied component $\mathcal{C} \in \mathcal{F}$, we find the number of active stars that satisfy $\mathcal{C}$. The objective is to find the maximum such number $\Delta^*_{\tilde{\rho}}$, over all unsatisfied components. First, each white node $v$ that centers an active star $X^i_v$ reports this star to each non-center node $u$ of it, by just sending $v_{id}$, special message \emph{active-star}, and the id of the component $\mathcal{C}$ for which $u$ is responsible for communicating with (if there is one). Then, for each white node $w$ and each unsatisfied component $\mathcal{C}$ that $w$ is responsible for communicating with it in any star, node $u$ sends to one of the nodes of $\mathcal{C}$ the number of stars in which $u$ is responsible for communicating with $\mathcal{C}$. These counts are summed up in each component $\mathcal{C}$ via Thurimella's algorithm, and it is called the \emph{active-degree} of $\mathcal{C}$. The maximum active-degree is found globally and called $\Delta^*_{\tilde{\rho}}$. 

\item[\textbf{($\mathcal{S}5$)}] Next, some active stars propose to their adjacent unsatisfied components.\fullOnly{

\hspace{3ex}}\shortOnly{}We mark each active star with probability $\frac{1}{5\Delta^*_{\tilde{\rho}}}$, where the decision is made randomly by the center of the star and sent to the other nodes of the star (if there is any). Then, these marks are sent to the components that get satisfied by the marked stars, as proposals, via the white nodes that are responsible for communicating with the components. If $v$ is self-sufficient, it would need to send at most one proposal to each adjacent component (it would be to those components for which $v$ is responsible for communicating with them in $X^i_v$). However, if $v$ is not self-sufficient, then $v$ might want to send many proposals to an unsatisfied component adjacent to it (there is at most one such component). This is not feasible in the \congest model. Instead, $v$ selects at most $3$ of these proposals (arbitrarily) and just submits these $3$ proposals.

\item[\textbf{($\mathcal{S}6$)}] Each component grants at most $3$ of the proposals it receives. This is done via Thurimella's algorithm, where $3$ proposals with largest center ids are granted. Finally components report the granted proposals to the adjacent white nodes.

\item[\textbf{($\mathcal{S}7$)}] Each marked star collects how many of its proposals are granted. If at least $1/3$ of the proposals of this star were granted, then all nodes of this marked star become gray. After that, we use Thurimella's algorithm again to identify the green nodes which their component (in $\mathcal{F}$) is satisfied (by checking if their component has a gray node). 

\item[\textbf{($\mathcal{S}8$)}] Finally, we have a \emph{clean up} step, which removes the newly-formed $\tilde{\rho}$-augmented basic-stars that if not removed now, might be active in the next iterations.  
Temporarily (just for this clean up step) color each white node \emph{blue} if its cost is at most $1/\tilde{\rho}$. For each unsatisfied component $\mathcal{C} \in \mathcal{F}$ that can be satisfied using only blue nodes, we find one or two blue nodes that connect $\mathcal{C}$ to some other component in $\mathcal{F}$ and we color these blue nodes gray, thus making $\mathcal{C}$ satisfied.
In the first round, for each blue node $v$, if $v$ is adjacent to only one component, it sends the id of this component and its own id $v_{id}$. If $v$ is adjacent to two or more components, it just sends its own id with an indicator symbol \emph{``two-or-more"}. 
In the second round, for each blue node $u$, if $u$ is adjacent to an unsatisfied component $\mathcal{C}$,  node $u$ creates a proposal for $\mathcal{C}$ as follows: if $u$ is adjacent to at least one other component $\mathcal{C}'\in \mathcal{F}$, then the proposal is simply the id of $u$. If $u$ is not adjacent to any other component $\mathcal{C}'$ but there is a blue neighbor $w$ of $u$ such that in the first round, $w$ sent the id of a component $\mathcal{C}''\neq \mathcal{C}$ or $w$ sent the \emph{``two-or-more"} indicator symbol, then the proposal contains the ids of $u$ and $v$. Otherwise, the proposal is empty. 
Each unsatisfied component picks one (nonempty) proposal, if it receives any, and grants it. The granted proposal is reported to all nodes adjacent to the component and if the proposal of $u$ is granted, it becomes gray and if this granted proposal contained a blue neighbor $w$, then $u$ informs $w$ about the granted proposal which means that $w$ also becomes gray. 
\end{itemize}

\fullOnly{
\paragraph{A remark about the time complexity} In the above algorithm, each phase takes $\mathcal{O}(\log^3 n)$ iterations, w.h.p., which leads to  \fullOnly{$\mathcal{O}((D+\sqrt{n}\log^* n)\log^3 n)$ rounds for each phase, and thus} $\mathcal{O}((D+\sqrt{n}\log^* n)\log^4 n)$ rounds for the whole algorithm. One can remove one logarithmic factor off of this complexity by (further) leveraging the fact that in each phase, we need to satisfy only half of the components. \fullOnly{To do that, if for a max-efficiency level $\tilde{\rho}$ and the respective max-component-degree $\Delta^*_{\tilde{\rho}}$, we have satisfied at least half of the components with degree at least $\Delta^*_{\tilde{\rho}}/2$, we can \emph{excuse} the other half from needing to be satisfied in this phase. This way, with constant probability, after just a constant number of iterations, we are done with components of degree at least $\Delta^*_{\tilde{\rho}}/2$. A standard concentration bound then shows that w.h.p. $\mathcal{O}(\log n)$ iterations are enough for all degree levels (with respect to efficiency $\tilde{\rho}$).} We defer the formal claim about this and the detailed algorithm to the journal version.}

\shortOnly{
\bigskip
\noindent Due to the space limitations, the analysis are deferred to the full version.}

\fullOnly{\input{Analysis}}
%\pagebreak
\medskip
%\paragraph{Discussions} See \Cref{app:open} for a discussion about the open problems and future work.

\fullOnly{
\section{Open Problems and Future Work}
\label{app:open}
This paper presents a distributed $\mathcal{O}(\log n)$ approximation algorithm for the \MCDS problem in $\tilde{\mathcal{O}}(D+\sqrt{n})$ rounds of the \congest model. 

As mentioned before, \MCDS is $\mathsf{NP}$-hard and if one assumes that nodes can only perform polynomial-time computations (which is a practically reasonable assumption and also a usual one\cite{Dubhashi-03, Jia-Rajmohan-Suel, Kuhn-Wattenhofer-03}), the $\mathcal{O}(\log n)$ approximation factor is optimal up to a constant factor, unless $\mathsf{P} = \mathsf{NP}$. The author finds it quite an intriguing question to see if one can get an $o(\log n)$ approximation in a non-trivial number of rounds, by relaxing this assumption about local computations. However, this question might be only of theoretical interest.

In the current presentation of the algorithm, we have not tried to optimize the constant in the approximation factor. However, it is not clear how to get a $(1+o(1))\log n$ approximation and that is another interesting question. 

The author started looking into the \MCDS problem with the hope of solving it---i.e., finding an $\mathcal{O}(\log n)$ approximation for it---in $\tilde{\mathcal{O}}(D+\sqrt{n})$ rounds of a more restricted version of the \congest model where in each round, each node can send one $\mathcal{O}(\log n)$-bits message to all of its neighbors. Notice that the same message has to be sent to all neighbors. This model is called $\mathsf{VCONGEST}$ as the congestion is on vertices, rather than on edges.  Note that this restriction is natural in node-capacitated networks, and \MCDS is also more important in such settings. It would be interesting to see if an $\tilde{\mathcal{O}}(D+\sqrt{n})$-rounds $\mathcal{O}(\log n)$-approximation for \MCDS can be found in $\mathsf{VCONGEST}$. 
}
 
\subsection*{Acknowledgment} We thank Fabian Kuhn for valuable discussions. We also thank Stephan Holzer and Christoph Lenzen for helpful comments about the presentation.

This work was supported by Simons award for graduate students in theoretical Computer Science (number 318723), AFOSR contract number FA9550-13-1-0042, NSF award 0939370-CCF, NSF award CCF-1217506, and NSF award CCF-AF-0937274.

%\clearpage

{
\shortOnly{\small}
\bibliographystyle{abbrv}
\bibliography{ref}
}

\fullOnly{
\appendix

\section{A Comment on the Unweighted Case of \MCDS}
\label{app:unweighted}
Here, we briefly comment that the unweighted case of \MCDS, where all nodes have equal weight, is quite different from the weighted case and it can be approximated locally. More precisely, an $\mathcal{O}(\log \Delta)$ factor approximation---where $\Delta$ is the maximum degree---exists, which uses only $\mathcal{O}(\log ^{2} n)$ rounds of the \congest model. 

Notice that similarly, the unweighted case of \MST is different from the weighted case in the sense that, a spanning connected selection of $O(n)$ edges (which is like a constant approximation of \MST in the unweighted case) can be found locally \footnote{Note that on the other hand, any approximation for a connected spanning subgraph in the weighted case needs $\tilde{\Omega}(D+\sqrt{n})$ rounds\cite{DasSarma-11}.}. More precisely, there is an algorithm that finds a set of $O(n)$ edges connecting the whole graph, in $\mathcal{O}(\log n \cdot 2^{\log^* n})$ rounds of the \congest model; this is the linear skeleton algorithm of Pettie\cite{Pettie-Skeleton}.

Going back to \MCDS in the unweighted case, let us first briefly recap on the algorithm of \cite{Dubhashi-03} which finds an $\mathcal{O}(\log \Delta)$ approximation in $\mathcal{O}(\log ^{2} n)$ rounds of the \local model: First, use the $\mathcal{O}(\log^{2} n)$ rounds dominating set  $\mathcal{O}(\log \Delta)$ approximation algorithm of \cite{Jia-Rajmohan-Suel}. Suppose that this dominating set is called $S$. Then, the remaining problem is to add $\mathcal{O}(\log \Delta)$ nodes to $S$ and achieve connectivity. Consider the virtual graph $\mathcal{H}_S = (S, E_S)$ on the dominating set $S$ where two nodes $v, u\in S$ are adjacent--- that is, $e=(v, u) \in E_S$---if they are within distance $3$ of each other in $G$. It is easy to see that this graph connected. The remaining task is to pick only $O(S)$ edges of this subgraph (each edge contains at most two nodes), while ensuring connectivity. For this, Dubhashi et al. follow the famous strategy that, if one destroys cycles of length $\mathcal{O}(\log n)$, graph has at most linear many edges. Destroying cycles of length $\mathcal{O}(\log n)$ is easy in the \local model; each node learns its $\mathcal{O}(\log n)$-neighborhood. Then, it throws away each of its incident edges if the edge has largest id (the id of an edge is formed by concatenating the ids of its two endpoint, the larger first) is the smallest edge-id in a cycle of length $2\log n+1$.

\begin{figure}[t]
	\centering
		\includegraphics[width=0.70\textwidth]{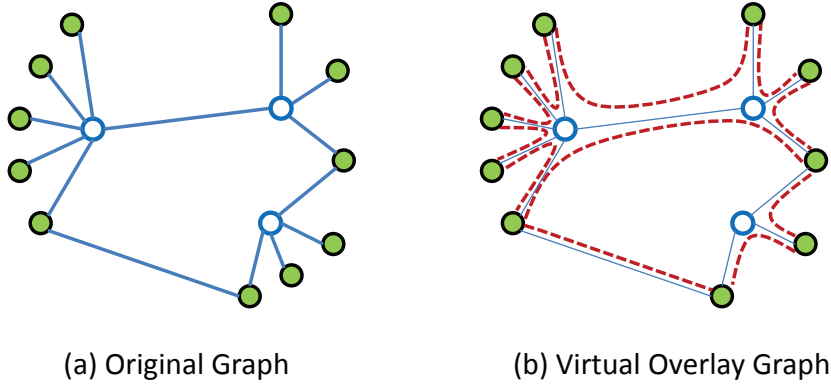}
	\caption{{\scriptsize The overlay graph between the dominating nodes. The edges of the original graph are presented as solid blue lines (on the left) and those of the overlay graph are presented as dashed red lines (on the right). }}
	\vspace{-10pt}
	\label{fig:Transformation}
\end{figure}

Now let us see how one can turn this idea to work in the \congest model. The dominating set approximation algorithm of Jia et al.\cite{Jia-Rajmohan-Suel} is already in the \congest model. Furthermore, one can substitute the simple \emph{cycle-destroying} part of \cite{Dubhashi-03} with the more sophisticated \emph{linear skeleton} algorithm of Pettie\cite{Pettie-Skeleton} that works in $\mathcal{O}(\log n \cdot 2^{\log^* n})$ rounds of the \congest model. The only remaining issue is that, we have to ensure that the \congest model allows us to run Pettie's algorithm on the virtual graph $\mathcal{H}_S$. This is not possible with the current definition of $\mathcal{H}_S$ as each edge of graph $G$ might be in many edges of $\mathcal{H}_S$. But the remedy is simple, we just make a small change in the definition of $\mathcal{H}_S$: For each node $w\in V\setminus S$, let $w$ order its neighbors in $S$ arbitrarily, say $s^w_1$, $s^w_2$, \dots, $s^w_\ell$. Now, define $E'_S$ as follows: for each $w\in V\setminus S$, add the following edges: $(s^{w}_1, s^{w}_2)$, $(s^{w}_2, s^{w}_3)$, \dots, $(s^{w}_{\ell-1}, s^{w}_{\ell})$. Furthermore, for each two $G$-neighboring nodes $w, w' \in V\setminus S$, which have $S$-neighbors respectively $s^w_1$, $s^w_2$, \dots, $s^w_\ell$ and $s^{w'}_1$, $s^{w'}_2$, \dots, $s^{w'}_{\ell'}$, we put two edges in $E'_S$: $(s^{w}_1, s^{w'}_1)$ and $(s^{w}_\ell, s^{w'}_{\ell'})$. We note that just one of these edges would be enough here but we add two to keep the symmetry. Also, if two nodes in $S$ are neighbors, put an $E'_S$-edge between them. Now define the new virtual graph to be simply $\mathcal{H}'_S=(S, E'_S)$. \Cref{fig:Transformation} shows a simple example. It is easy to see that each edge of this virtual graph goes through at most two nodes of $G$, $\mathcal{H}'_S$ is connected. Furthermore, each edge of $G$ is used in at most two edges of $\mathcal{H}'_S$. Hence, each communication round on $\mathcal{H}'_S$ can be simulated by two communication rounds on $G$. Therefore, the issue of congestion on the virtual graph is fixed.

\section{Round Complexity Lower Bound}
\label{app:LB}
Here, we mention the simple observation that the techniques of Das Sarma et al.~\cite{DasSarma-11} imply a $\tilde{\Omega}(D+\sqrt{n})$ rounds lower bound for any approximation of \MCDS. For simplicity, we only explain an $\Omega(\sqrt{n}/\log n)$-round lower bound for the case where $D=\mathcal{O}(\log n)$ and in fact we will just sketch the changes. We encourage the interested reader to see \cite{DasSarma-11} for the details and generalization to other diameter values.

\begin{observation} For any polynomial $\alpha(n)$, there is a constant $\eps>0$ such that any $\alpha(n)$-approximation algorithm for the minimum-weight connected dominating set that has error-probability at most $\eps$ requires at least $\Omega(\sqrt{n}/\log n)$ rounds of the \congest model, on a graph that has diameter $D=\mathcal{O}(\log n)$.
\end{observation}

The general approach (following \cite{DasSarma-11}) is that we present a graph and show that one can encode instances of two-party set disjointness on $\sqrt{n}$-bits in the node-weights of this graph such that the following holds: if there is an $\alpha(n)$-approximation algorithm $\mathcal{A}$ for the minimum-weight connected dominating set problem that has error-probability at most $\eps$ and uses $T \leq \Omega(\sqrt{n}/\log n)$ rounds, then there is a randomized algorithm for two party set on instances with $\sqrt{n}$-bits inputs with error-probability at most $\eps$ that uses less than $\Theta(T \log n)$ communication rounds, which would be a contradiction. Without loss of generality, enhance $\mathcal{A}$ so that each node knows the total weight of final CDS, note that this can be done in additional $\mathcal{O}(D)$ rounds and is thus without loss of generality.

\shortOnly{
\begin{figure}[b!]
	\centering
		\includegraphics[width=0.75\textwidth]{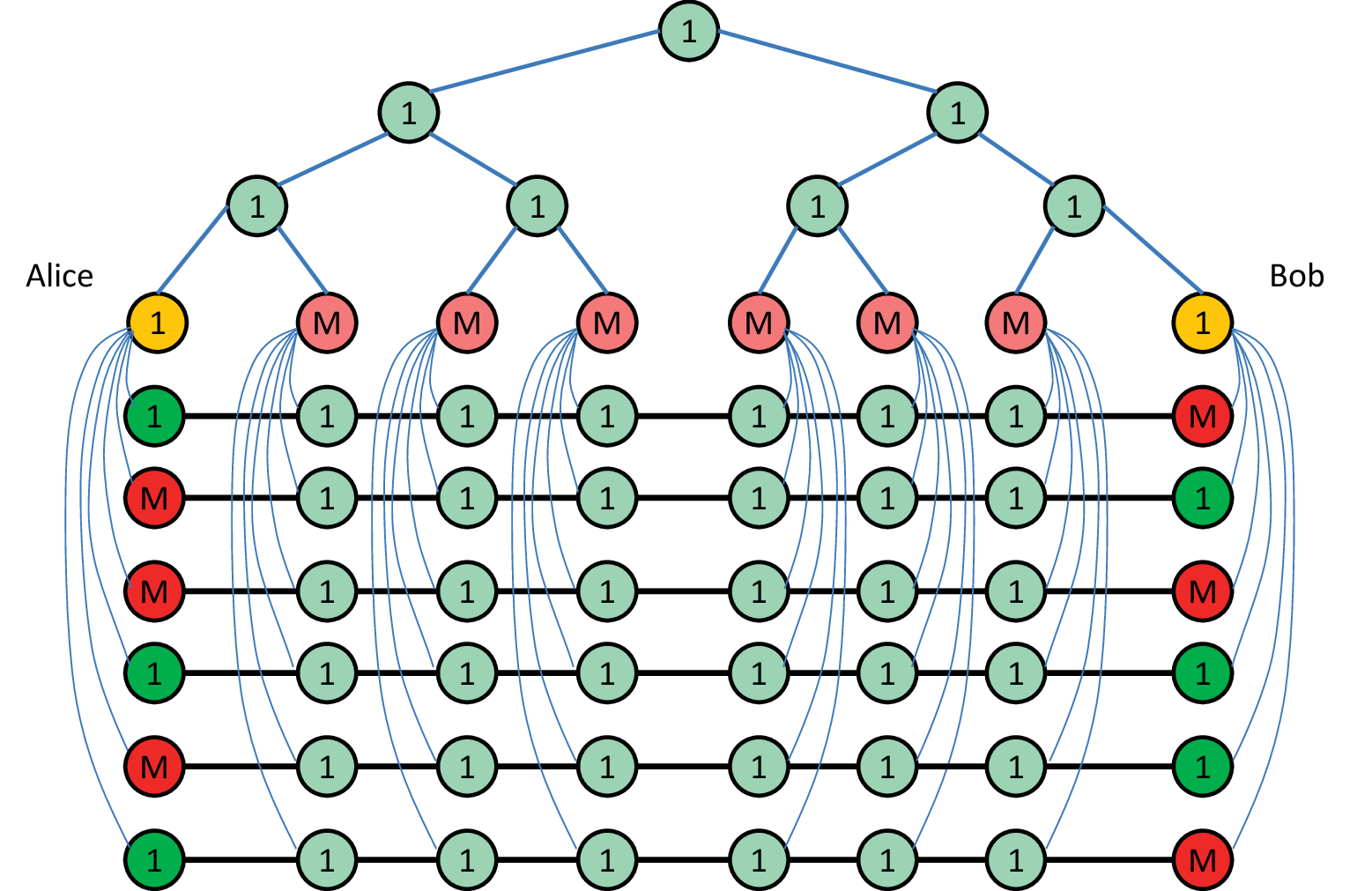}
	\caption{{\scriptsize The round complexity lower bound graph}}
	\label{fig:LB}
\end{figure}
}
\fullOnly{
\begin{figure}[h!]
	\centering
		\includegraphics[width=0.75\textwidth]{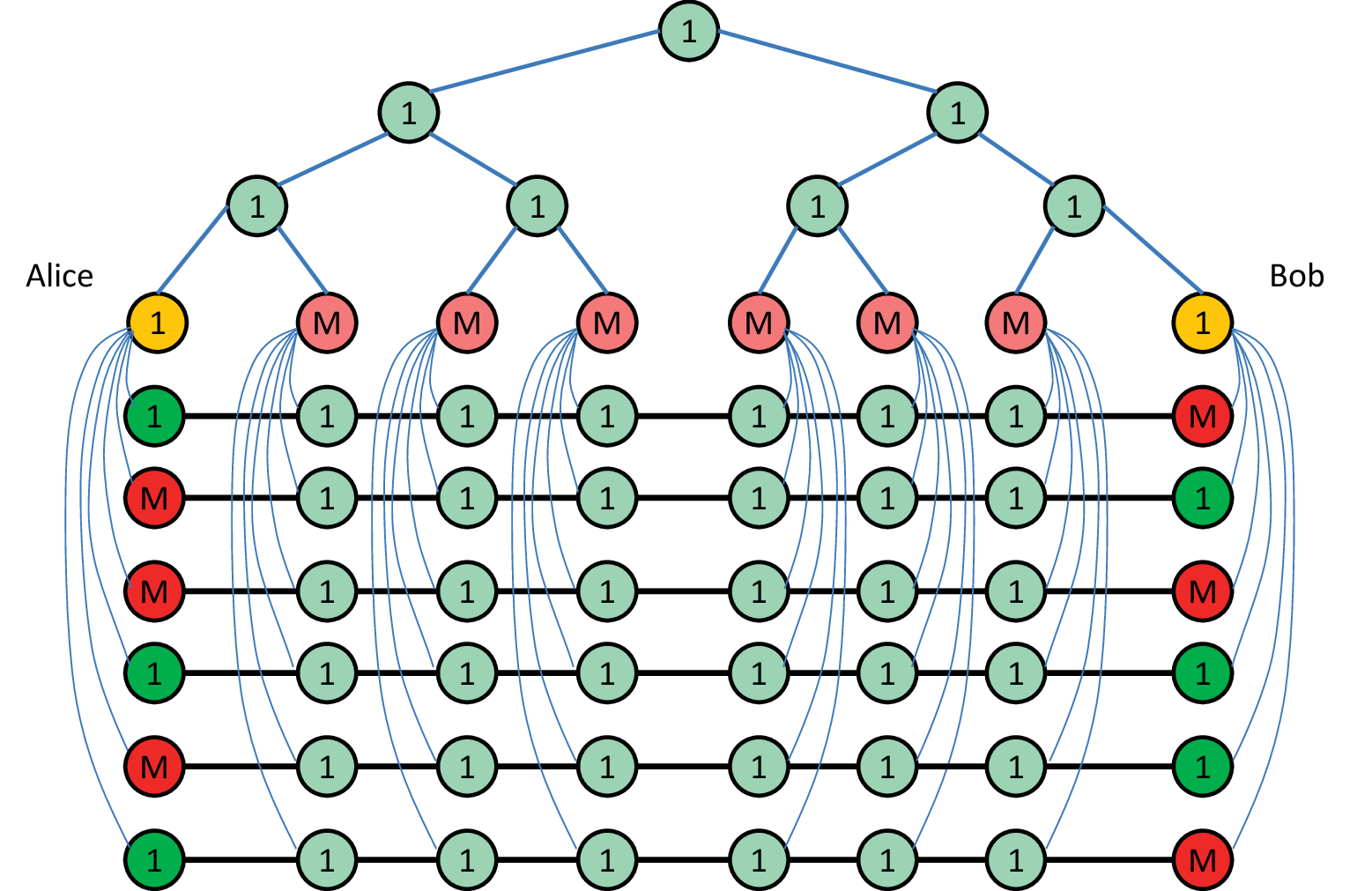}
	\caption{{\scriptsize The round complexity lower bound graph}}
	\label{fig:LB}
\end{figure}
}
Consider a graph made of three parts: $\sqrt{n}$ aligned parallel paths of length $\sqrt{n}$ each, a tree of depth $\log n -1$ on top of these trees such that each of its leaves is aligned with one column of the nodes of the paths, and finally, for each leaf of the tree, edges from this leaf to all the nodes in the paths that are in the same column. \Cref{fig:LB} shows an example.

Call the left-most leaf Alice and the rightmost leaf Bob and that they are given an instance of set-disjointness, where Alice and Bob respectively get inputs $\mathcal{X}$ and $\mathcal{Y}$ that are subsets of $\{1, 2, \dots, \sqrt{n}\}$. We next describe how to encode these inputs in the weights of the \MCDS problem. Give a weight of $1$ to each non-leaf node of the tree, the nodes held by Alice and Bob, and also all nodes on the paths except the leftmost and the rightmost ones on each path. These weight-$1$ nodes are indicated with a light green color in \Cref{fig:LB}. Then, for each other leaf node, give a weight $M=\alpha(n) \cdot n+1$ (light red color). Finally, for each $i\in \{1, \dots, \sqrt{n}\}$, for the leftmost node of the $i^{th}$ path, give a weight of $1$ if $i\notin\mathcal{X}$ (dark green color) and give a weight $M$ if $i\in \mathcal{X}$ (dark red). Similarly, for the rightmost node of the $i^{th}$ path, give a weight of $1$ if $i\notin\mathcal{Y}$ and give a weight $M$ if $i\in \mathcal{Y}$. Note that Alice and Bob can indeed put these weights as inputs to the \MCDS problem in just one round, as all the nodes know the fixed part of the weight, and the variable part which depends on the set disjointness inputs is on the neighbors of Alice and Bob, and thus, they can lean their weights in just one round.

Now notice that if $i\in \mathcal{X}\cap \mathcal{Y}$, then any CDS must contain at least one node of weight $M$. On the other hand, if $\mathcal{X}$ and $\mathcal{Y}$ are disjoint, then there is an CDS with weight (less than) $n$, which includes all weight-$1$ nodes. Since $\frac{M}{n}>\alpha$, and as $\mathcal{A}$ finds an $\alpha(n)$-approximation of \MCDS, the weight output by $\mathcal{A}$ lets the nodes distinguish the case where the sets are disjoint from the case where they are not (with error-probability being the same as in $\mathcal{A}$). The final piece, which is the key technical part, is to show that Alice and Bob can indeed simulate $\mathcal{A}$ being run over the whole graph, using only $\Theta(T\log n)$ communication rounds between themselves. This follows exactly from \cite[Simulation Theorem]{DasSarma-11}.

}

\end{document}

%% file: Analysis.tex
\subsection{Analysis}
\label{app:analysis}
For the analysis, we need to establish two facts, (1) that the cost of each phase is $\mathcal{O}(\mathsf{OPT})$, and (2) that each phase takes only $\mathcal{O}(\log^3 n)$ iterations, w.h.p. As each iteration is implemented in $\mathcal{O}(D+\sqrt{n}\log^* n)$ rounds, these prove the desired properties of each phase.

\subsubsection{Cost Related Analysis}
\label{subsub:cost}
 In each iteration, we color some white nodes gray and thus satisfy some components. This is done in a way that the overall efficiency of the nodes added in this iteration is within a constant factor of the best basic-star. That is, the number of components satisfied in this iteration is $\Theta(\rho^*)$ times the total cost of the nodes colored gray in this iteration. As the heart of cost analysis, we show that in each iteration, as long as at least $N/2$ unsatisfied components exist, $\rho^*\geq \frac{N}{4\mathsf{OPT}}$. This will be done by showing that one can cover the (white nodes of) $\mathsf{OPT}$ with basic-stars, such that each white node is in at most $2$ basic-stars.

\begin{lemma}\label{lem:efficiency-LB} In each phase, we spend a cost of at most $\mathcal{O}(\mathsf{OPT})$.
\end{lemma}
\begin{proof} 
First note that, in steps $\mathcal{S}6$ and $\mathcal{S}7$, each active star has efficiency at least $\tilde{\rho}/2$, and an active star becomes gray if at least $1/3$ of its proposals are granted and each component grants at most $3$ proposals. Thus, the efficiency of the whole set of white nodes colored gray in step $\mathcal{S}7$ is $\Theta(\tilde{\rho})$. Moreover, in the clean up step (step $\mathcal{S}8$), each component grants at most one proposal and each proposal contains at most two blue nodes, each of which has weight at most $1/\tilde{\rho}$. Hence, the efficiency in the clean up step is also $\Theta(\tilde{\rho})$.

Now as the key part of the proof, we claim that in each iteration in which at least $N/2$ unsatisfied components remain, there is at least one basic-star that has efficiency of at least $\frac{N}{4\mathsf{OPT}}$. This claim implies that in this iteration, $\tilde{\rho} \geq \rho^*/2 \geq \frac{N}{8\mathsf{OPT}}$. Over all iterations of this phase, we satisfy at most $N$ components, always with an efficiency $\Theta(\tilde{\rho})$, which means that we spend a cost of $\mathcal{O}(\mathsf{OPT})$ over the whole phase. Recall that if the number of the unsatisfied components drops below $N/2$, we call the phase finished, and move to the next phase.

Now to prove the claim, consider one iteration and assume that at least $N/2$ unsatisfied components remain. For the sake of contradiction, suppose that each basic-star has efficiency strictly less than $\frac{N}{4\mathsf{OPT}}$. Consider the minimum-cost CDS $\mathsf{OPT}$. Let $T$ be the set of white nodes in $\mathsf{OPT}$. We cover $T$ with a number of basic-stars $X_1, X_2, \dots, X_\ell$ such that each node of $T$ is in at most two of these basic-stars and each unsatisfied component can be satisfied by at least one of these basic stars. Then, for each $X_i$, define $C'(X_i)=\Cost(X_i)/2$. Note that $\sum_{i=1}^{\ell} C'(X_i) \leq \sum_{v\in T} c(v)=\mathsf{OPT}$. Each basic-star $X_i$ splits cost $C'(X_i)$ equally between the unsatisfied components $\Phi(X)$ that get satisfied by $X$. That is, each such component gets cost $\frac{C'(X_i)}{|\Phi(X_i)|} > \frac{2\mathsf{OPT}}{N}$ from star $X_i$. Hence, each unsatisfied component gets a cost strictly greater than $\frac{2\mathsf{OPT}}{N}$ and summed up over all the unsatisfied components---which are at least $N/2$ many---, we get that $\sum_{i=1}^{\ell} C'(X_i)> \mathsf{OPT}$, which is a contradiction. 

What is left is thus to show that we can cover $T$ with a number of basic-stars $X_1$, $X_2$, $\dots$, $X_\ell$ such that each node of $T$ is in at most two of these basic-stars and each unsatisfied component can be satisfied by at least one of these basic stars. We give a simple sequential procedure which produces such basic-stars. During this procedure, each node $v\in T$ keeps a Boolean variable $hit_v$ which is false initially. For each node $v\in T$, call $v$ \emph{lonely} if it is adjacent to exactly one component and that component is not satisfied. 

Sequentially, go over the nodes in $T$ one by one and for each $v\in T$, do as follows: consider the star $X_v$ made of $v$ and all lonely neighbors $w$ of $v$ that are not hit so far, i.e., those such that $hit_w=false$. Add $X_v$ to the collection if it satisfies at least one component, and if this happens, also for each $w\in X_v\setminus \{v\}$, set $hit_w=true$. Note that if a lonely node $w$ gets hit, then the single unsatisfied component $\mathcal{C}$ adjacent to $w$ gets satisfied by $X_v$.

Now note that in this algorithm, each node $u$ will be in at most two stars, one star $X_u$ that is centered on $u$, and one star $X_v$ that is centered on a neighbor $v$ of $u$ and such that in the iteration in which we consider $w$, we set $hit_u=true$. 

On the other hand, consider an unsatisfied component $\mathcal{C}$. We show that $\mathcal{C}$ gets satisfied by one of the basic-stars produced by the above algorithm. Note that $\mathsf{OPT}$ satisfies $\mathcal{C}$. Therefore, there is a white node $v\in T$ that is adjacent to $\mathcal{C}$ and either $v$ is adjacent to a different component $\mathcal{C}'\neq \mathcal{C}$ or $v$ has another white neighbor $w \in T$ and $w$ is adjacent to a different component $\mathcal{C}'\neq \mathcal{C}$. Now if the node $v$ is not lonely, it is adjacent to at least two components, and hence $X_v$ satisfies $\mathcal{C}$ and we are done. Otherwise, suppose $v$ is lonely. Since $v$ is lonely, when we consider $w$ in the loop, either $v$ is already hit by some other basic-star $X_{w'}$, or the basic-star $X_w$ hits $v$. In either case, $\mathcal{C}$ gets satisfied. This finishes the proof. 
\end{proof}

\subsubsection{Speed Related Analysis}
\label{subsub:speed}
The speed analysis has more subtle points. We show that after $\mathcal{O}(\log^3 n)$ iterations, at least half of the components would be satisfied and thus this phase ends. A critical point for establishing this is to show that, thanks to the clean up step (analyzed in \Cref{lem:cleanup}), for each unsatisfied component, the number of active $\tilde{\rho}$-augmented basic-stars $X^i_v$ that satisfy this component is monotonically non-increasing when viewed as a function of the iteration number $i$. This part will be our main tool for managing the issue of ``\emph{new stars}" (discussed in Challenge 2 above), and is proven in \Cref{lem:decreasing-degrees}. Furthermore, in \Cref{lem:progress}, we show that with at least a constant probability, a constant fraction of the components that are now each in at least $\Delta^*_{\tilde{\rho}}/2$ active stars $X^i_v$ of iteration $i$ get satisfied. Hence, it will follow that in iteration $j=i+\mathcal{O}(\log n)$, there remains no unsatisfied component that can be satisfied by at least $\Delta^*_{\tilde{\rho}}/2$ many active stars $X^j_v$, w.h.p. Thus, in each $\mathcal{O}(\log n)$ iterations, $\Delta^*_{\tilde{\rho}}$ decreases by a factor of $2$, w.h.p. After $\mathcal{O}(\log^2 n)$ iterations of the loop, there will be no $\tilde{\rho}$-augmented basic-star, and hence, no basic-star with efficiency $\tilde{\rho}$. Then we move to the next efficiency level, which is at most $\tilde{\rho}/2$. After $\mathcal{O}(\log^3 n)$ iterations, more than half of the components would be satisfied and we stop the phase.

\begin{lemma}\label{lem:cleanup}In the clean up step, if a unsatisfied component $\mathcal{C}$ could have been satisfied using only blue nodes, it indeed gets satisfied.
\end{lemma}
\begin{proof}
Suppose that unsatisfied component $\mathcal{C}$ can be satisfied using only blue nodes. Then, as green-or-gray nodes dominate the graph, there is a component $\mathcal{C}'$ and either one blue node $v$ such that  $\mathcal{C}$-$v$-$\mathcal{C}'$ connects $\mathcal{C}$ to $\mathcal{C}'$or two adjacent blue nodes $v$ and $w$ such that $\mathcal{C}$-$v$-$w$-$\mathcal{C}'$ is a path connecting $\mathcal{C}$ to $\mathcal{C}'$. In the former case, $v$ clearly proposes to $\mathcal{C}$. In the latter case, $v$ will receive either the id of $\mathcal{C}'$ or the special symbol ``\emph{two-or-more}" from $w$. And thus again, in either case, $v$ proposes to $\mathcal{C}$. Therefore, component $\mathcal{C}$ will receive at least one proposal. Component $\mathcal{C}$ will accept one proposal, and this will make it connect to one other component $\mathcal{C}''$ (which might be equal to $\mathcal{C}'$) and hence satisfied.
\end{proof}

\begin{lemma}\label{lem:decreasing-degrees} For each unsatisfied component $\mathcal{C}\in \mathcal{F}$ and the almost-max-efficiency $\tilde{\rho}$, the number of $\tilde{\rho}$-augmented basic-stars $X^i_v$ selected in \emph{step $\mathcal{S}3$} that satisfy $\mathcal{C}$ does not increase from one iteration to the next.
\end{lemma}
\begin{proof} Fix an iteration $i\geq 2$. We claim that if an unsatisfied component $\mathcal{C}\in \mathcal{F}$ can be satisfied by a $\tilde{\rho}$-augmented basic-star $X^i_v$ centered on $v$ that was selected  in step $\mathcal{S}3$ of iteration $i$, then there was one $\tilde{\rho}$-augmented basic-star $X^{i-1}_v$ centered on $v$ that was selected in step $\mathcal{S}3$ of iteration $i-1$ and $\mathcal{C}$ could be satisfied by $X^{i-1}_v$ as well. This then directly leads to the lemma.

We first show that it cannot be the case that the core $X''$ of $X^i_v$ (see step $\mathcal{S}3$ and Definition \ref{def:aug-star} for definition of core) was a useless star in iteration $i-1$. Having this established, we then show that $X^{i-1}_v$ could satisfy $\mathcal{C}$, as well. 

First, for the sake of contradiction, suppose that the core basic-star $X''\subseteq X^i_v$ was useless in iteration $i-1$. Suppose that $X''$ could satisfy unsatisfied component $\mathcal{C}'$ in iteration $i$. Component $\mathcal{C}'$ might be equal to $\mathcal{C}$ or not. As $X''$ is useful in iteration $i$ but not in iteration $i-1$, it means that there is a node $u$ that was white at the start of iteration $i-1$ and it became gray at the end of that iteration and such that $u$ connects $X''$ to a now satisfied component $\mathcal{C}''\neq \mathcal{C}'$. 
%
%
%Furthermore in iteration $i-1$, $X''$ cannot be adjacent to more than one component as otherwise $X''$ would be useful for satisfying $\mathcal{C}'$ in iteration $i-1$. Similarly,
%
In iteration $i$, $X''$ cannot be adjacent to two unsatisfied components as then it would be useful in iteration $i-1$. As in iteration $i$ basic-star $X''$ is $\tilde{\rho}$-efficient and it satisfies only one unsatisfied component, we get that the total cost of nodes in $X''$ is at most $\frac{1}{\tilde{\rho}}$. Hence, all nodes of $X''$ were blue in the clean up step of iteration $i-1$. Furthermore, $u$ was either gray at the start of the clean up step of iteration $i-1$ or it was a blue node in that step and then it became gray. We know that $u$ is not adjacent to $\mathcal{C}'$ (otherwise $\mathcal{C}'$ would be satisfied). But, we know that each node of $X''$ must be adjacent to at least one green node, and in iteration $i-1$, $X''$ could not have been adjacent to more than one component (otherwise it would be useful for satisfying $\mathcal{C}'$). Thus, we get that each node of $X''$ is adjacent to component $\mathcal{C}'$. Therefore, $\mathcal{C}'$ could have been satisfied using only one or two blue nodes: either with one blue node of $X''$ connecting it to $u$ which was gray then, or with one blue node of $X''$ and node $u$ which was blue then. Hence, \Cref{lem:cleanup} gives that $\mathcal{C}'$ must have been satisfied at the end of iteration $i-1$ (perhaps through a different path). This is in contradiction with $X''$ having $\mathcal{C}'$ as its unsatisfied adjacent component in iteration $i$. Thus, we conclude that the core basic-star $X''$ was useful in iteration $i-1$.

Note that if an unsatisfied component is adjacent to $X''$ in iteration $i$, it was adjacent to $X''$ in iteration $i-1$ as well. Hence, the number of unsatisfied components that could be satisfied by $X''$ in iteration $i-1$ is at least as many as those that could be satisfied in iteration $i$. This establishes that $X''$ was at least $\tilde{\rho}$-efficient in iteration $i-1$. Thus, indeed there was a $\tilde{\rho}$-augmented basic-star $X^{i-1}_v$ centered on $v$ and selected in step $\mathcal{S}3$ of iteration $i-1$. It remains to show that $X^{i-1}_v$ could satisfy $\mathcal{C}$.

For the sake of contradiction, suppose that $X^{i-1}_v$ was not adjacent to $\mathcal{C}$ (as otherwise we would be done). It means that there is another white node $w\in X^i_v\setminus X^{i-1}_{v}$ that connects $\mathcal{C}$ to $v$. Also, $\mathcal{C}$ is the only component in $\mathcal{F}$ that is adjacent to $w$ as otherwise, $w$ would have been self-sufficient and hence it would not report $\mathcal{C}$ to $v$ in iteration $i$ and thus it would not be in $X^{i}_v$ (recall the definition of basic-star). Therefore, we know that in iteration $i-1$, $\{w\}$ could have potentially been a good auxiliary-leg for the core of $X^{i-1}_v$. As $\{w\}$ was not included in $X^{i-1}_v$, we know $\{w\}$ was not a good auxiliary-leg. As $X^{i-1}_v$ is not adjacent to $\mathcal{C}$, from \Cref{def:aug-star} we can infer it must have been the case that $\Cost(w)>2/\tilde{\rho}$. But now in iteration $i$, basic-star $X^i_v$ which includes $w$ and thus has cost strictly greater than $2/\tilde{\rho}$ has efficiency $\tilde{\rho}$. So, $X^i_v$ must satisfy at least $3$ components. But then, even if we discard $w$ from star $X^i_v$, we get a smaller $\tilde{\rho}$-efficient basic-star. Hence, $w$ was not included in the core $X''$ of $X^i_v$, which means that $w$ was included in $X^i_v$ as a good auxiliary-leg, showing that $\Cost(w)\leq 2/\tilde{\rho}$, which is a contradiction. Having arrived at the contradiction from the assumption that $X^{i-1}_v$ was not adjacent to $\mathcal{C}$, we get that $X^{i-1}_v$ must indeed have been adjacent to $\mathcal{C}$. That is, the $\tilde{\rho}$-augmented basic-star centered on $v$ in iteration $i-1$ could have satisfied $\mathcal{C}$. This completes the proof. 
\end{proof}

\begin{lemma}\label{lem:progress} In each iteration $i$, the set of grayed stars has efficiency within a constant factor of the max-efficient basic-star. Furthermore, with at least a constant probability, a constant fraction of the components that can be satisfied by at least $\Delta^*_{\tilde{\rho}}/2$ many $\tilde{\rho}$-augmented basic-stars $X^i_v$ get satisfied.
\end{lemma}
\begin{proof}
For the first part, note that a basic-star joins if it is almost max-efficient and at least $1/3$ of its proposals are granted, and each component grants at most $3$ proposals. It follows that the set of grayed stars has efficiency within a constant factor of the max-efficient basic-star.
 
For the second part, first note that the probability that an almost-max-efficient star $X$ is marked active and at least $1/3$ of its proposals are accepted is $\Theta(1/\Delta^*)$. This is because, $X$ is marked active with probability $1/(5\Delta^*)$ and then, for each unsatisfied component $\mathcal{C}$ that gets satisfied by $X$, the probability that more than $3$ stars satisfying $\mathcal{C}$ are marked is at most $\binom{\Delta^*_{\tilde{\rho}}}{4} (\frac{1}{5\Delta^*_{\tilde{\rho}}})^4\leq (\frac{e}{5})^4< 1/10$. Hence, the expected fraction of the unaccepted proposals of $X$ is at most $1/10$, which using Markov's inequality means that the probability that more than $2/3$ are unaccepted is at most $3/20$. Therefore, the probability that $X$ is marked active and at least $1/3$ of its proposals are accepted is at least $0.03/\Delta^*$.

Call a component \emph{large-degree} if it can be satisfied by at least $\Delta^*_{\tilde{\rho}}/2$ many almost-max-efficient star stars. We get that for each large-degree unsatisfied component $\mathcal{C}$, the expected number of stars that satisfy $\mathcal{C}$ and get colored gray is at least $1/100$. On the other hand, the probability that there are $z$ stars that satisfy $\mathcal{C}$ and are colored gray (which shows that they are marked active) decays exponentially with $z$, as it is at most $\binom{\Delta^*_{\tilde{\rho}}}{z} (\frac{1}{5\Delta^*_{\tilde{\rho}}})^z\leq (\frac{e}{5})^z$.  It follows that with at least a constant probability, one or more of stars that satisfy $\mathcal{C}$ gets colored gray. This is because otherwise, only an $\eps$ of the total probability mass is on $z\geq 1$, for a sub-constant $\eps$, which given the exponentially decaying tail, it would contradict with the expectation being at least constant $1/100$. Hence, we get that $\mathcal{C}$ gets satisfied with at least a constant probability.  

It follows from an application of Markov's inequality that with at least a constant probability, at least a constant fraction of large-degree components get satisfied, finishing the proof.
\end{proof}